\documentclass{LMCS}

\def\dOi{10(4:14)2014}
\lmcsheading%
{\dOi}
{1--21}
{}
{}
{May\phantom.23, 2014}
{Dec.~23, 2014}
{}

\ACMCCS{[{\bf Theory of computation}]: Logic}


\usepackage{hyperref}

\usepackage{graphicx}

\usepackage{amsmath}
\usepackage{comment}
\usepackage{amssymb}

\usepackage{alltt}

\theoremstyle{plain}

\begin{document}

\title[Beyond Q-Resolution and Prenex Form]{Beyond Q-Resolution and Prenex Form:\\
A Proof System for
Quantified Constraint Satisfaction}

\author[Hubie Chen]{Hubie Chen}	
\address{
Departamento LSI,
Universidad del Pa\'{i}s Vasco,
E-20018 San Sebasti\'{a}n,
Spain
\emph{and}
IKERBASQUE, Basque Foundation for Science,
E-48011 Bilbao,
Spain
}	







\begin{abstract}
We consider the quantified constraint satisfaction problem (QCSP)
which is to decide,
given a structure and a first-order sentence (not assumed here to be in
prenex form) built from conjunction and quantification,
whether or not the sentence is true on the structure.
We present a proof system for certifying the falsity of QCSP instances
and develop its basic theory; for instance, we provide an algorithmic
interpretation of its behavior.  Our proof system places 
the established Q-resolution proof system in a broader context,
and also allows us to derive QCSP tractability results.
\end{abstract}

\maketitle

\newtheorem{example}[thm]{Example}

\newtheorem{definition}[thm]{Definition}
\newtheorem{lemma}[thm]{Lemma}
\newtheorem{theorem}[thm]{Theorem}

\newcommand{\ppequiv}{\mathsf{PPEQ}}
\newcommand{\eq}{\mathsf{EQ}}
\newcommand{\iso}{\mathsf{ISO}}
\newcommand{\ppeq}{\ppequiv}
\newcommand{\ppiso}{\mathsf{PPISO}}
\newcommand{\boolppiso}{\mathsf{BOOL}\mbox{-}\mathsf{PPISO}}
\newcommand{\csp}{\mathsf{CSP}}
\newcommand{\gi}{\mathsf{GI}}
\newcommand{\ci}{\mathsf{CI}}

\newcommand{\rela}{\mathbf{A}}
\newcommand{\relb}{\mathbf{B}}
\newcommand{\relc}{\mathbf{C}}
\newcommand{\rels}{\mathbf{S}}
\newcommand{\relt}{\mathbf{T}}

\newcommand{\alga}{\mathbb{A}}
\newcommand{\algb}{\mathbb{B}}
\newcommand{\algab}{\mathbb{A}_{\relb}}

\newcommand{\idemp}{I}

\newcommand{\nats}{\mathbb{N}}

\newcommand{\varv}{\mathcal{V}}
\newcommand{\variety}{\mathcal{V}}
\newcommand{\false}{\mathsf{false}}
\newcommand{\true}{\mathsf{true}}
\newcommand{\pol}{\mathsf{Pol}}
\newcommand{\inv}{\mathsf{Inv}}
\newcommand{\alg}{\mathsf{Alg}}
\newcommand{\pitwo}{\Pi_2^p}
\newcommand{\sigmatwo}{\Sigma_2^p}
\newcommand{\pithree}{\Pi_3^p}
\newcommand{\sigmathree}{\Sigma_3^p}

\newcommand{\fancya}{\mathcal{A}}
\newcommand{\fancyc}{\mathcal{C}}
\newcommand{\fancyg}{\mathcal{G}}
\newcommand{\fancym}{\mathcal{M}}
\newcommand{\tree}{\mathcal{T}}

\newcommand{\tw}{\mathsf{tw}}

\newcommand{\qc}{\mathsf{QC\mbox{-}MC}}
\newcommand{\rqc}{\mathsf{RQC\mbox{-}MC}}

\newcommand{\qcfo}{\mathrm{QCFO}}
\newcommand{\qcfofk}{\qcfo_{\forall}^k}
\newcommand{\qcfoek}{\qcfo_{\exists}^k}

\newcommand{\fo}{\mathrm{FO}}

\newcommand{\tup}[1]{\overline{#1}}

\newcommand{\nn}{\mathsf{nn}}
\newcommand{\bush}{\mathsf{bush}}
\newcommand{\width}{\mathsf{width}}

\newcommand{\un}{N^{\forall}}
\newcommand{\en}{N^{\exists}}

\newcommand{\ord}{\tup{u}}
\newcommand{\ordp}[1]{\tup{#1}}

\newcommand{\gc}{G^{-C}}

\newcommand{\HOM}{\mathsf{HOM}}
\newcommand{\HOMP}[1]{#1\mbox{-}\HOM}

\newcommand{\free}{\mathsf{free}}

\newcommand{\todo}[1]{{\bf To do: } #1}

\newcommand{\core}{\mathsf{core}}

\newcommand{\arrow}{\rightarrow}
\newcommand{\arrowp}[1]{\stackrel{#1}{\arrow}}
\newcommand{\arrowk}{\arrowp{k}}

\renewcommand{\S}{\mathcal{S}}
\newcommand{\ar}{\mathsf{ar}}

\newcommand{\id}{\mathsf{id}}

\newcommand{\G}{G}

\newcommand{\hcomment}[1]{{\bf Hubie comment: } #1 {\bf End}}

\newcommand{\res}{\upharpoonright}

\newcommand{\vars}{\mathsf{vars}}


\section{Introduction}

\paragraph{\bf Background.}
The study of \emph{propositional proof systems}
 for 
certifying the unsatisfiability of quantifier-free propositional
formulas
is supported by multiple motivations~\cite{BeamePitassi98-survey,Segerlind07-survey}.  
First, the desire to have an efficiently verifiable certificate 
of a formula's unsatisfiability is a natural and basic one,
and indeed the field of propositional proof complexity
studies, for various proof systems, whether and when
succinct proofs exist for unsatisfiable formulas.
Next, theorem provers are typically based on such proof
systems,
and so insight into the behavior of proof systems
can yield insight into the behavior of theorem provers.
Also, algorithms that perform search to determine the satisfiability
of formulas can typically be shown to implicitly generate proofs in a proof
system,
and thus lower bounds on proof size translate to lower bounds
on the running time of such algorithms.
Finally, 
algorithms that check for unsatisfiability proofs of various
restricted forms have been shown to yield
tractable cases of the propositional satisfiability problem
and related problems
(see for example~\cite{AtseriasKolaitisVardi04-propagation-as-proof-system,BartoKozik09-boundedwidth}).

In recent years, increasing attention has been directed towards
the study of \emph{quantified proof systems} that certify the falsity of
quantified propositional formulas, which study is also pursued
with the motivations similar to those outlined for the quantifier-free case.
Indeed, the development of so-called \emph{QBF solvers}, which
determine the truth of quantified propositional formulas,
has become an active research theme, and the study of
quantified proof systems is pursued
as a way to understand their behavior, as well as to
explore the space of potential certificate formats 
for verifying their correctness on particular input 
instances~\cite{NarizzanoPeschieraPulinaTacchella09-qbfs}.

\emph{Q-resolution}~\cite{BuningKarpinskiFlogel95-resolution}
is a quantified proof system that can be viewed as 
a quantified analog of \emph{resolution},
one of the best-known and most customarily considered propositional proof
systems.
In the context of quantified propositional logic,
Q-resolution is a heavily studied and basic proof system
on which others are built and to which others are
routinely compared, as well as a point of departure for
the discussion of suitable certificate formats
(see~\cite{GiunchigliaNarizzanoTacchella06-clauseterm,SlivovskySzeider12-pathdependencies,JanotaMarquesSilva13-expansions}
for examples).

However, the Q-resolution proof system has intrinsic shortcomings.
First, it is only applicable to quantified propositional sentences that
are in prenex form, that is, where all quantifiers appear in front.
While it is certainly true that an arbitrary 
given quantified propositional formula may be
efficiently prenexed, the process of prenexing is not canonical:
intuitively, it involves choosing a total order of variables consistent with 
the partial order given by the input formula.
As argued by Egly, Seidl, and Woltran~\cite{EglySeidlWoltran09-solver-nnf},
this may disrupt the original formula structure,
``artificially extend'' the scopes of quantifiers,
and generate dependencies among variables that were not originally
present, unnecessarily increasing the expense of solving; 
we refer the reader to their article for a contemporary discussion of
this issue.\footnote{Let us remark that 
using so-called \emph{dependency schemes} is a potential way
to cope with such introduced dependencies 
in a prenex formula~\cite{SlivovskySzeider12-pathdependencies}.}
A second shortcoming of Q-resolution is that
it is only defined in the propositional setting,
despite that some scenarios may be more naturally and cleanly modelled
by allowing variables to be quantified over domains of size greater
than two.

\paragraph{\bf Contributions.}
In this article, we introduce a proof system 
that 
directly overcomes both of the identified shortcomings of Q-resolution
and
that, in a sense made precise,
 generalizes Q-resolution.

We here define 
the \emph{quantified constraint satisfaction problem (QCSP)}
to be the problem of deciding,
given a relational structure $\relb$ and a first-order sentence $\phi$
(not necessarily in prenex form)
built from the
conjunction connective ($\wedge$)  and the two quantifiers ($\forall$,
$\exists$),
whether or not the sentence is true on the structure.
To permit different variables to have different domains, we
formalize the QCSP using multi-sorted first-order logic.

Our proof system (Section~\ref{sect:qcsp-proof-system})
allows for 
the certification of falsity of QCSP instances.
While Q-resolution provides rules for deriving clauses
from a given quantified propositional formula,
our proof system provides rules for deriving what we call
\emph{constraints}
at various formula locations of a given QCSP instance;
here, a constraint $(V, F)$ is a set $V$ of variables paired with 
a set $F$ of assignments, each defined on $V$.
A formula location $i$ paired with a constraint 
is called a \emph{judgement};
a proof in our system is a sequence of judgements
where each is derived from the previous ones via the rules.

Crucially, we formulate and prove
a key lemma (Lemma~\ref{lemma:judgement-gives-formula})
that shows (essentially) that if a judgement $(i, V, F)$ is derivable
from a QCSP instance $(\phi, \relb)$,
then there exists a formula $\psi(V)$ that ``defines''
the constraint $(V, F)$ over $\relb$,
such that $\psi(V)$ can be conjoined to the input sentence $\phi$ at
location $i$ while preserving logical equivalence.
This key lemma is then swiftly deployed to establish
soundness and completeness of our proof system
(Theorem~\ref{thm:sound-and-complete}).
We view the formulation of our proof system and of this key lemma
as conceptual contributions.
They offer a broader, deeper,
and more general perspective on Q-resolution and what it means
for a clause to be derivable by Q-resolution:
we show (in a sense made precise) 
that each clause derivable by Q-resolution is derivable
by our proof system (see Theorem~\ref{thm:simulation-of-qres}). 
This yields a clear and transparent proof of the soundness
of Q-resolution which, interestingly, 
is carried out in the setting of first-order logic, 
despite the result concerning propositional logic.

In order to relate our proof system to Q-resolution,
we give a proof system for certain quantified propositional formulas
(Section~\ref{sect:prop-proof-system})
and prove that this second proof system is a faithful propositional
interpretation of our QCSP proof system
(Theorem~\ref{thm:relationship}).
We also provide an algorithmic interpretation of this second proof
system.
In particular, we give a nondeterministic search algorithm 
such that traces of execution that result in certifying falsity
correspond to refutations in the proof system
(Section~\ref{subsect:algorithm}).
As a consequence, the proof system yields 
a basis for establishing running-time lower bounds for 
any deterministic algorithm which instantiates
the non-deterministic choices of our search algorithm.

In the final section of the article (Section~\ref{sect:consistency}),
we present and study a notion of \emph{consistency} for the QCSP
that is naturally induced by our proof system.
In the context of constraint satisfaction,
a consistency notion is a condition
which is necessary for the satisfiability of an instance
and which can typically be efficiently checked.
An example used in practice is \emph{arc consistency},
and understanding when various forms of
consistency provide an exact characterization of satisfiability
(that is, when consistency is sufficient for satisfiability in
addition to being necessary)
has been a central theme in the tractability theory
of constraint
satisfaction~\cite{AtseriasKolaitisVardi04-propagation-as-proof-system,BartoKozik09-boundedwidth,ChenDalmauGrussien13-ACandfriends}.
Atserias, Kolaitis, and Vardi~\cite{AtseriasKolaitisVardi04-propagation-as-proof-system}
showed that checking for
$k$-consistency,
an oft-considered consistency notion,
can be viewed as detecting the absence of a proof of unsatisfiability
having width at most $k$, in a natural proof system
(the width of a proof is the maximum number of variables appearing in
a line of the proof);
Kolaitis and Vardi~\cite{KolaitisVardi00-gametheoretic} 
also characterized $k$-consistency
algebraically as whether or not \emph{Duplicator} can win
a natural \emph{Spoiler-Duplicator pebble game} 
in the spirit of 
Ehrenfeucht-Fra\"{i}ss\'{e} games.

Inspired by these connections, 
we directly define 
a QCSP instance to be \emph{$k$-judge-consistent} 
if it has no unsatisfiability proof (in our proof system)
of width at most $k$; and,
we then present an algebraic, 
Ehrenfeucht-Fra\"{i}ss\'{e}-style characterization
of $k$-judge-consistency (Theorem~\ref{thm:consistency-characterization}).
As an application of this algebraic characterization,
we prove that (in a sense made precise) 
any case of the QCSP that lies in the tractable regime
of a recent dichotomy theorem~\cite{ChenDalmau12-decomposingquantified},
is tractable via checking for $k$-judge-consistency.\footnote{
  Let us remark that this dichotomy theorem has since been generalized~\cite{Chen14-frontier}.
}
That is, within the framework considered by that dichotomy,
if a class of QCSP instances is tractable at all, it is tractable via
$k$-judge consistency.
We remark that earlier work~\cite{ChenDalmau05-pebblegames} presents algebraically
a notion of consistency for the QCSP,
but no corresponding proof system was explicitly presented;
the notion of $k$-judgement consistency 
can be straightforwardly verified to imply
the notion of consistency in this earlier article.

To sum up, this article presents
a proof system for non-prenex quantified formulas.
Our proof system is based on highly natural and simple rules,
and its utility is witnessed by its connections to Q-resolution
and by our presentation of a consistency notion that it induces,
which allows for the establishment of tractability results.
We hope that this proof system will serve as a point of reference and foundation for
the future study of solvers and certificates for non-prenex formulas.
One particular possibility for future work is to compare
this proof system to others that are defined on non-prenex formulas,
such as those discussed and studied by Egly~\cite{Egly12-sequent}.

\section{Preliminaries}

When $f$ is a mapping, we use $f \res S$ to denote its restriction to
a set $S$; 
we use $f[s \to b]$ to denote the extension of $f$ that
maps $s$ to $b$.
We extend this notation to sets of mappings in the natural fashion.

\paragraph{\bf First-order logic.}
We assume basic familiarity with the syntax and semantics of
first-order logic.
For the sake of broad applicability, in this article,
we work with multi-sorted relational first-order logic,
formalized here as follows.
A \emph{signature} is a pair $(\sigma, \S)$
where $\S$ is a set of \emph{sorts} and $\sigma$ is a set of relation
symbols;
each relation symbol $R \in \sigma$ has an associated arity $\ar(R)$
which is an element of $\S^*$.
Each variable $v$ has an associated sort $s(v) \in \S$;
an \emph{atom} is a formula
$R(v_1, \ldots, v_k)$ where $R \in \sigma$
and $s(v_1) \ldots s(v_k) = \ar(R)$.
A \emph{structure} $\relb$ on signature $(\sigma, \S)$
consists of an $\S$-indexed family $B = \{ B_s ~|~ s \in \S \}$
of sets, called the \emph{universe} of $\relb$,
and, for each symbol $R \in \sigma$,
an interpretation $R^{\relb} \subseteq B_{\ar(R)}$.
Here, for a word $w = w_1 \ldots w_k \in \S^*$,
we use $B_w$ to denote the product $B_{w_1} \times \cdots \times
B_{w_k}$.
Suppose that $V$ is a set of variables where each variable $v \in V$
has an associated sort $s(v)$; 
by a mapping $f: V \to B$, we mean a mapping that sends each
$v \in V$ to an element $f(v) \in B_{s(v)}$.

When $\phi$ is a formula, we use $\free(\phi)$ to denote the set
containing the free variables of $\phi$.
The \emph{width} of a formula $\phi$ is the maximum of $|\free(\psi)|$
over all subformulas $\psi$ of $\phi$.
A \emph{quantified-conjunctive formula} (for short, \emph{qc-formula})
is a formula over a signature built from atoms on the signature,
conjunction ($\wedge$), and the two quantifiers ($\forall$, $\exists$).
Note that we permit conjunction of arbitrary arity.
As expected, a \emph{qc-sentence} is a qc-formula $\phi$
such that $\free(\phi) = \emptyset$.
We allow conjunction of arbitrary (finite) arity, in formulas.
We will use $\top$ to denote a sentence that is always true;
this is considered to be a qc-sentence.
A relation $P$ is \emph{qc-definable} over a structure $\relb$
if there exists a qc-formula
$\phi(v_1, \ldots, v_k)$ such that
$P \subseteq B_{s(v_1) \ldots s(v_k)}$ and 
$P$ contains a tuple $(b_1, \ldots, b_k)$
if and only if $\relb, b_1, \ldots, b_k \models \phi(v_1, \ldots, v_k)$.
Note that by the notation
$\relb, b_1, \ldots, b_k \models \phi(v_1, \ldots, v_k)$,
we mean that the structure $\relb$ and the
mapping taking each $v_i$ to $b_i$ satisfy $\phi$.
We sometimes use $\equiv$ to indicate logical equivalence of two
formulas.

We define the \emph{QCSP} to be the problem of deciding,
given a \emph{QCSP instance}, 
which is a pair $(\phi, \relb)$ where $\phi$ is a qc-sentence
and $\relb$ is a structure that both have the same signature,
whether or not $\relb \models \phi$.

\section{QCSP proof system}
\label{sect:qcsp-proof-system}

In this section, we present our proof system for the QCSP,
and establish some basic properties thereof, including soundness and completeness.
Let $(\phi, \relb)$ be a QCSP instance, and conceive of $\phi$ as a tree.
The proof system will allow us to derive what we call
\emph{constraints}
at the various nodes of the tree.
To facilitate the discussion,
we will assume that each qc-sentence $\phi$ 
has, associated with it, a set $I_{\phi}$ of \emph{indices}
that contains one index for each subformula occurrence of $\phi$,
that is, for each node of the tree corresponding for $\phi$.
Let us remark that (in general) the
collection of constraints derivable
at an occurrence of a subformula does not depend only on the
subformula,
but also on the subformula's location in the full formula $\phi$.
When $i$ is an index, we use $\phi(i)$ to denote the 
actual subformula of the subformula occurrence corresponding to $i$;
we will also refer to $i$ as a \emph{location}.

\begin{figure*}[t]
\label{fig:formula}
\centering
\includegraphics[width=150pt]{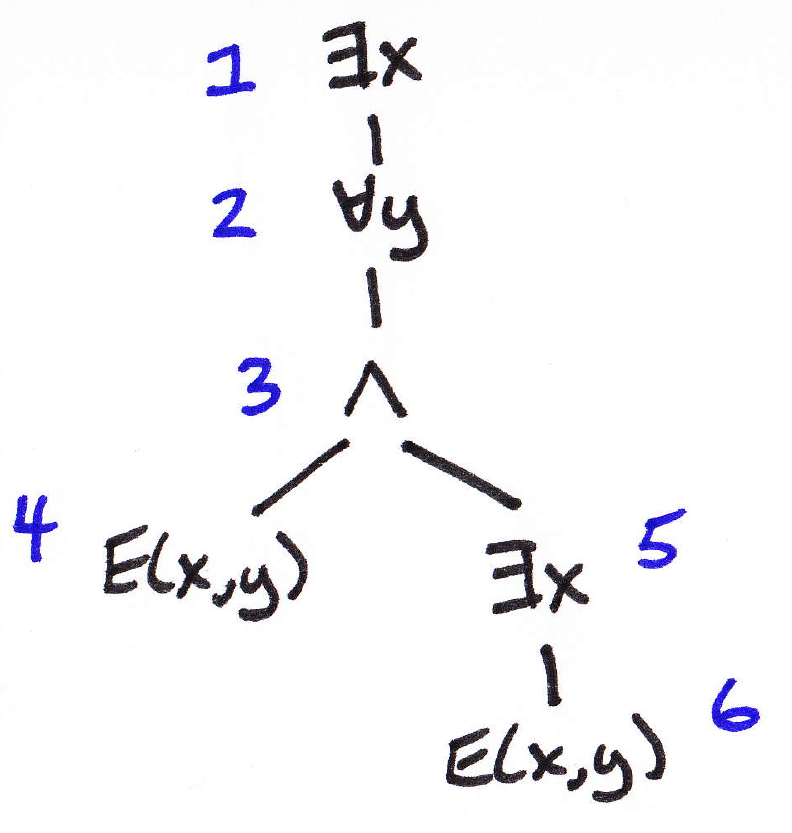}
\caption{Formula discussed in Examples~\ref{ex:running}
and~\ref{ex:proof}.}
\end{figure*}

\begin{example} \label{ex:running}
\normalfont
Consider the qc-sentence 
$\phi = \exists x \forall y (E(x,y) \wedge (\exists x E(x,y)))$.
(See Figure~\ref{fig:formula}.)
When viewed as a tree, this formula has $6$ nodes.
We may index them naturally according to the depth-first search order:
we could take the index set $\{ 1, \ldots, 6 \}$
where $\phi(4) = \phi(6) = E(x,y)$,
$\phi(5) = \exists x \phi(6)$,
$\phi(3) = \phi(4) \wedge \phi(5)$,
$\phi(2) = \forall y \phi(3)$,
and
$\phi(1) = \exists x \phi(2)$.
\qed \end{example}

We say that an index $i$ is a \emph{parent} of an index $j$,
and also that $j$ is a \emph{child} of $i$,
if, in viewing the formula $\phi$ as a tree, the root of the
subformula occurrence of $i$ is the parent of the
root of the subformula occurrence of $j$.
Note that, when this holds, the formula $\phi(i)$
either is of the form 
$Q v \phi(j)$ where $Q$ is a quantifier and $v$ is a variable,
or is a conjunction where $\phi(j)$ appears as a conjunct.
As examples,
with respect to the qc-sentence and indexing in
Example~\ref{ex:running},
index $3$ has two children, namely, $4$ and $5$,
and index $3$ has one parent, namely, $2$.

\begin{definition}
\normalfont
Let $(\phi, \relb)$ be a QCSP instance.
A \emph{constraint} (on $(\phi, \relb)$) is a pair $(V, F)$
where $V$ is a set of variables occurring in $\phi$,
and $F$ is a set of mappings from $V$ to $B$.
A \emph{judgement}
(on $(\phi, \relb)$) is a triple
$(i, V, F)$ where $i \in I_{\phi}$ and
$(V, F)$ is a constraint with $V \subseteq \free(\phi(i))$;
it is \emph{empty} if $F = \emptyset$.
\end{definition}

Here, we use the convention that 
(relative to a QCSP instance)
there is exactly one map
$e: \emptyset \to B$ defined on the empty set,
so there are two constraints
whose variable set is the empty set:
the constraint $(\emptyset, \emptyset)$, 
and the constraint $(\emptyset, \{ e \})$
where $e$ is the aforementioned map.

When $(U_1, F_1)$, $(U_2, F_2)$ are two constraints
on the same QCSP instance, we define
the \emph{join} of $F_1$ and $F_2$,
denoted by $F_1 \Join F_2$, to be
the set
\begin{center}
$\{ f: U_1 \cup U_2 \to B ~|~ 
(f \res U_1) \in F_1, (f \res U_2) \in F_2 \}.$
\end{center}
When $(U, F)$ is a constraint and $y$ is a variable in $U$,
we use $\epsilon_y F$ to denote the set
\begin{center}
$\{ f: U \setminus \{ y \} \to B ~|~ 
\textup{for each $b \in B_{s(y)}$, it holds that $f[y \to b] \in F$}
\}.$
\end{center}
The operator $\epsilon_y$ will be used to eliminate
a universally quantified variable $y$.
Dually, in the following definition, \emph{projection}
can be used to cope with existential quantification.

\renewcommand{\arraystretch}{0.1}
\newcommand{\tinyskip}{\vspace{1pt}}

\begin{definition} \label{proof-system}
\normalfont
A \emph{judgement proof}
on $(\phi, \relb)$
is a finite sequence of judgements,
each of which has one of the following types:

\tinyskip

\begin{tabular}{lp{8.5cm}}

(atom) & $(i, \{v_1, \ldots, v_k\}, F)$
\tinyskip

 where $\phi(i)$ is an atom $R(v_1, \ldots, v_k)$ and

$F = \{ f: \{ v_1, \ldots, v_k \} \to B ~|~ (f(v_1), \ldots, f(v_k))
\in R^{\relb} \}$

\\

(projection) &  $(i, U, F \res U)$
\tinyskip

where $(i, V, F)$
is a previous judgement,
and $U \subseteq V$

\\

(join) & $(i, U_1 \cup U_2, F_1 \Join F_2)$
\tinyskip

where $(i, U_1, F_1)$ and $(i, U_2, F_2)$ are previous judgements

\\

(upward flow) & $(i, V, F)$
\tinyskip

where $(j, V, F)$
is a previous judgement and

$i$ is the parent of $j$

\\


($\forall$-elimination) & $(i, V \setminus \{ y \}, \epsilon_y F)$
\tinyskip

where $(j, V, F)$ is a previous judgement with $y \in V$,

$\phi(i) = \forall y \phi(j)$, and $i$ is the parent of $j$

\\

(downward flow) & $(j, V, F)$
\tinyskip

where $(i, V, F)$ is a previous judgement
and 

$i$ is the parent of $j$

\\

\end{tabular}

We say that a judgement $(i, V, F)$ is \emph{derivable} if there exists
a judgement proof that contains the judgement.

The \emph{width} of a judgement $(i, V, F)$ is $|V|$.
The \emph{width} of a judgement proof is the maximum width
over all of its judgements, and the \emph{length} of a judgement proof
is the number of judgements that it contains.
\qed \end{definition}

Let us emphasize that, by definition, a judgement proof
is a finite sequence of judgements, and
by definition, in order for a triple $(i, V, F)$
to be a judgement, it must hold that all variables in $V$
are free variables of $\phi(i)$.  
Consequently, \emph{upward flow} can 
only be applied to a judgement $(j, V, F)$
if all variables in $V$ are free variables of $\phi(i)$,
where $i$ is the parent of $j$; 
an analogous comment holds for \emph{downward flow}.

\begin{example}
\label{ex:proof}
\normalfont
Let $\phi$ be the qc-sentence from Example~\ref{ex:running}
(shown in Figure~\ref{fig:formula}),
considered as a sentence over signature
$(\{ E \}, \{ e, u \})$ with $\ar(E) = eu$ and
where $s(x) = e$ and $s(y) = u$.
Define $\relb$ to be a structure over this signature
having universe $B$ defined by
$B_e = \{ a, b, c \}$ and $B_u = \{ d, e, f \}$,
and where
$E^{\relb} = \{ (a, d), (a, e), (a, f), (b, e), (c, f) \}$.
To offer a feel of the proof system,
we give some examples of derivable judgements.

Let $F_E$ be the set of assignments from $\{ x, y \}$ to $B$
that satisfy $E(x,y)$ (over $\relb$).
By (atom), we may derive the judgement
$(4, \{ x, y \}, F_E)$.
By (upward flow), we may then derive the judgement
$(3, \{ x, y \}, F_E)$.
By ($\forall$-elimination), we may then derive the judgement
$(2, \{ x \}, G)$,
where $G$ contains the single map that takes $x$ to $a$.
By applying (downward flow) twice, we may then derive the judgement
$(4, \{ x \}, G)$.
By (atom), we may also derive the judgement
$(6, \{ x, y \}, F_E)$.
By (projection), we may then derive the judgement
$(6, \{ x \}, H)$, where $H$ contains the maps taking
$x$ to $a$, $b$, and $c$, respectively.
Let us remark that, even though $\phi(4) = \phi(6)$
and we derived the judgement $(4, \{ x \}, G)$,
it is not possible to derive the judgement 
$(6, \{ x \}, G)$.
(This can be verified 
by Lemma~\ref{lemma:judgement-gives-formula}, to be presented next,
and the observation that $\relb \models \phi$.)
\qed \end{example}

\newcommand{\towedge}{\stackrel{\wedge}{\to}}

We now prove soundness and completeness of our proof system;
we first establish a lemma, which indicates what it means
for a judgement to be derivable.

When $\phi$ is a qc-formula with index set $I$,
and $\{ \theta_i \}_{i \in I}$ is a family of formulas,
we use $\phi^{+\theta}$ to denote the formula obtained
from $\phi$ by replacing, at each location $i$, 
the subformula $\phi(i)$ by $\phi(i) \wedge \theta_i$.
Formally, we define $\phi^{+\theta}$ by induction.
When $\phi(i)$ is an atom, we define
$\phi^{+\theta}(i) = \phi(i) \wedge \theta_i$.
When $\phi(i) = \phi(j) \wedge \phi(k)$,
we define
$\phi^{+\theta}(i) = \phi^{+\theta}(j) \wedge \phi^{+\theta}(k) \wedge
\theta_i$.
When $\phi(i) = Qv \phi(j)$, 
we define
$\phi^{+\theta}(i) = (Q v \phi^{+\theta}(j)) \wedge \theta_i$.
We define $\phi^{+\theta}$ to be $\phi^{+\theta}(r)$ where $r$ is the
root index of $\phi$ (that is, where $r$ is such that $\phi(r) = \phi$).

\begin{lemma}
\label{lemma:judgement-gives-formula}
Let $(\phi, \relb)$ be a QCSP instance.
For every derivable judgement $(i, V, F)$,
there exists a qc-formula $\psi$ such that 
\begin{itemize}

\item $\free(\psi) = V$;

\item for each $f: V \to B$, it holds that
$f \in F$ if and only if $\relb, f \models \psi$; and

\item for any family $\{ \theta_i \}_{i \in I}$ of formulas,
it holds that
$\phi^{+\theta}$ entails $\phi^{+\theta}[i \towedge \psi]$
(and hence that 
$\phi^{+\theta} \equiv \phi^{+\theta}[i \towedge \psi]$).
%
%
Here, $\phi^{+\theta}[i \towedge \psi]$ denotes the formula
where, at location $i$, 
the subformula 
$\phi^{+\theta}(i)$
is replaced with 
$\phi^{+\theta}(i) \wedge \psi$.
\end{itemize}
Moreover, if one has a judgement proof of width at most $k$,
then each of the formulas $\psi$ produced for its judgements
has $\width(\psi) \leq k$.
\end{lemma}

\begin{proof}
We consider the different types of judgements, and use the notation
from Definition~\ref{proof-system}.
In each case, the claim on the width is straightforwardly verified.

In the case of (atom), we take $\psi = \phi(i)$;
the formulas
$\phi^{+\theta}$ and
$\phi^{+\theta}[i \towedge \psi]$
are logically equivalent since $\phi(i) \equiv \phi(i) \wedge
\phi(i)$.

In the case of (projection), 
by induction, we have that 
$\phi^{+\theta}$ 
entails
$\phi^{+\theta}[i \towedge \psi']$ 
where $\psi'$ is the formula
for the judgement 
$(i, V, F)$.
We take $\psi = \exists v_1 \ldots \exists v_m \psi'$,
where $v_1, \ldots, v_m$ is a listing of the elements in $V \setminus
U$.  We have that 
$\phi^{+\theta}[i \towedge \psi']$ 
entails 
$\phi^{+\theta}[i \towedge \psi]$,
and hence by transitivity of the entailment relation that
$\phi^{+\theta}$
entails 
$\phi^{+\theta}[i \towedge \psi]$.

In the case of (join),
by induction, we have 
(for any family $\{ \theta_j \}_{j \in I}$)
that 
$\phi^{+\theta}$ 
entails both
$\phi^{+\theta}[i \towedge \psi_1]$ 
and
$\phi^{+\theta}[i \towedge \psi_2]$ 
where $\psi_1$ and $\psi_2$ are the formulas corresponding to 
the judgements $(i, U_1, F_1)$ and $(i, U_2, F_2)$.
We take $\psi = \psi_1 \wedge \psi_2$.
Fix a family $\{ \theta_j \}_{j \in I}$, and
define $\{ \theta'_j \}_{j \in I}$ to be the family that
has $\theta'_i = \theta_i \wedge \psi_1$, and is everywhere else
equal to $\{ \theta_j \}_{j \in I}$.
We have that $\phi^{+\theta}$ entails
$\phi^{+\theta}[i \towedge \psi_1] \equiv \phi^{+\theta'}$,
and that
$\phi^{+\theta'}$ entails $\phi^{+\theta'}[i \towedge \psi_2]$.
It follows that $\phi^{+\theta}$ entails
$\phi^{+\theta'}[i \towedge \psi_2] \equiv 
\phi^{+\theta}[i \towedge (\psi_1 \wedge \psi_2)]$.

In the case of ($\forall$-elimination),
by induction, we have that
$\phi^{+\theta}$ 
entails 
$\phi^{+\theta}[j \towedge \psi']$,
where $\psi'$ is the formula for the judgement $(j, V, F)$.
We take $\psi = \forall y \psi'$.
%
We claim that the formula 
$\phi^{+\theta}[j \towedge \psi']$
is logically equivalent to
$\phi^{+\theta}[i \towedge \psi]$,
which suffices to give that
$\phi^{+\theta}$ 
entails 
$\phi^{+\theta}[i \towedge \psi]$.
This is because the subformula of
$\phi^{+\theta}[j \towedge \psi']$
at location $i$ is logically equivalent to
$(\forall y (\phi^{+\theta}(j) \wedge \psi')) \wedge \theta_i$
which is logicically equivalent to 
$(\forall y \phi^{+\theta}(j)) \wedge (\forall y \psi') \wedge \theta_i$.

In the cases of (upward flow) and (downward flow),
we take $\psi$ to be equal to 
the formula that is given to us by the previous judgement.
It is straightforwardly verified that
$\phi^{+\theta}[i \towedge \psi]$ 
and
$\phi^{+\theta}[j \towedge \psi]$ 
are logically equivalent.
 \end{proof}

\begin{theorem}
\label{thm:sound-and-complete}
Let $(\phi, \relb)$ be a QCSP instance.
An empty judgement on $(\phi, \relb)$ is derivable
if and only if $\relb \not\models \phi$.
\end{theorem}

This theorem is proved in the following way.
The forward direction follows immediately from the previous lemma.
For the backward direction, we show by induction that,
for each location $i$, there exists a derivable judgement
for which the formula $\psi$ given by the previous lemma is
equal to $\phi(i)$!

\begin{proof}
Suppose that an empty judgement
$(i, V, F)$ is derivable.
Then by invoking the (projection) rule, the empty judgement
$(i, \emptyset, \emptyset)$ is derivable.
Define
$\{ \theta_j \}_{j \in I}$ so that $\theta_j$ is the true formula $\top$ for each $j \in I$;
then, invoking Lemma~\ref{lemma:judgement-gives-formula},
we have that there exists a qc-formula $\psi$ that is false
on $\relb$ (that is, $\relb \not\models \psi$)
and such that
$\phi$ entails $\phi[i \towedge \psi]$.
We have that $\relb \not\models \phi[i \towedge \psi]$,
and hence 
(since $\phi$ entails $\phi[i \towedge \psi]$)
we have that $\relb \not\models \phi$.

Suppose that $\relb \not\models \phi$.
We claim that for each location $i$,
there exists a derivable judgement $(i, V, F)$
where the corresponding formula $\psi$,
given by the proof of Lemma~\ref{lemma:judgement-gives-formula},
is equal to $\phi(i)$.
This suffices, as then the root location $r$ has a derivable judgement
$(r, V, F)$
such that $F = \emptyset$.
We establish the claim by induction.

When $\phi(i)$ is an atom, we use the judgement 
given by (atom) in the proof system
(Definition~\ref{proof-system}).
When $\phi(i)$ is a conjunction, let $j$ and $k$ be the children of
$i$,
so that $\phi(i) = \phi(j) \wedge \phi(k)$.
Let $(j, V_j, F_j)$ and $(k, V_k, F_k)$ be the derivable judgements
given by induction.
By (upward flow) in the proof system, we have
that $(i, V_j, F_j)$ and $(i, V_k, F_k)$ are derivable judgements;
by invoking (join), we obtain the desired judgement.
When $\phi(i)$ begins with existential quantification,
let $j$ be the child of $i$, and denote
$\phi(i) = \exists x \phi(j)$.
Let $(j, V, F)$ be the derivable judgement given by induction;
by applying the rule (projection) 
to obtain a constraint on $V \setminus \{ x \}$ 
and then the rule (upward flow), we obtain the
desired derivation.
When $\phi(i)$ begins with universal quantification,
let $j$ be the child of $i$, and denote
$\phi(i) = \forall y \phi(j)$.
Let $(j, V, F)$ be the derivable judgement
given by
induction;
by the ($\forall$-elimination) rule, we obtain the desired derivation.
\end{proof}

\section{Propositional proof system}
\label{sect:prop-proof-system}

\newcommand{\overnot}[1]{\overline{#1}}

\emph{In this section, we introduce a different proof system,
which is a propositional interpretation of the QCSP proof system.
For differentiation, we refer to judgements and judgement proofs
as defined in the previous section as constraint judgements
and constraint judgement proofs.}

A \emph{literal} is a propositional variable $v$ or the negation
$\overnot{v}$ thereof.
Two literals are \emph{complementary}
if one is a variable $v$ and the other is $\overnot{v}$;
each is said to be the \emph{complement} of the other.
A \emph{clause} is a disjunction of literals that 
contains, for each variable, at most one literal on the variable;
a clause is sometimes viewed as the set of the literals that
it contains.  
A clause is \emph{empty} if it does not contain any literals.
The variables of a clause are simply the variables that
underlie
the clause's literals, and the set of variables of a clause $\alpha$
is denoted by $\vars(\alpha)$.
A clause $\gamma$ is 
a \emph{resolvent} of two propositional clauses $\alpha$ and $\beta$
if there exists a literal $L \in \alpha$ such that its complement $M$
is in $\beta$, and
$\gamma = (\alpha \setminus \{ L \}) \cup (\beta \setminus \{ M \})$.
A clause $\gamma$ is \emph{falsified} by a propositional assignment
$a$
if $a$ is defined on $\vars(\gamma)$ and each literal in $\gamma$
evalutes to false under $a$.

We define a \emph{QCBF instance} to be a propositional formula
not having free variables that is built from clauses,
conjunction, and universal and existential quantification
over propositional variables.
As with QCSP instance, we assume that each QCBF instance $\psi$
has an associated index set that contains an index for each subformula
of $\psi$.
Note that a clause is not considered to have any subformulas, other
than itself.
As an example, consider the QCBF instance
$\exists x \forall y \exists z 
((\overnot{y} \vee z) \wedge (y \vee \overnot{z} \vee x))$.
This formula would have $6$ indices: one for each of the two clauses,
one for the conjunction of the two clauses, and one for each of the 
quantifiers.

Let $\psi$ be a QCBF instance.
A \emph{clause judgement} (on $\psi$) is a pair $(i, \alpha)$
where $i \in I_{\psi}$ and $\alpha$ is a clause with
$\vars(\alpha) \subseteq \free(\psi(i))$; a clause judgement $(i, \alpha)$
is \emph{empty} if $\alpha$ is empty.

\begin{definition} \label{def:clause-judgement-proof}
\normalfont
A \emph{clause judgement proof} on a QCBF instance $\psi$
is a finite sequence of clause judgements, 
each of which has one of the following types:

\tinyskip

\begin{tabular}{lp{8.5cm}}

(clause) & $(i, \alpha)$
\tinyskip

 where $\phi(i)$ is the clause $\alpha$

\\

(resolve) &  $(i, \gamma)$
\tinyskip

where $(i, \alpha)$ and $(i, \beta)$ are previous clause judgements,
and $\gamma$ is a resolvent of $\alpha$ and $\beta$

\\

(upward flow) & $(i, \alpha)$
\tinyskip

where $(j, \alpha)$
is a previous clause judgement and

$i$ is the parent of $j$

\\

($\forall$-removal) & $(i, \alpha \setminus \{ y, \overnot{y} \})$
\tinyskip

where $(j, \alpha)$ is a previous clause judgement,

$\phi(i) = \forall y \phi(j)$, and $i$ is the parent of $j$

\\

(downward flow) & $(j, \alpha)$
\tinyskip

where $(i, \alpha)$ is a previous clause judgement
and 

$i$ is the parent of $j$

\\

\end{tabular}

We say that a clause judgement $(i, \alpha)$ is \emph{derivable}
if there exists a clause judgement proof that contains the clause
judgement.
\qed \end{definition}

The \emph{width} of a clause judgement $(i, \alpha)$ is
$|\vars(\alpha)|$.
The \emph{width} of a clause judgement proof is the maximum width
over all of its clause judgements; 
the \emph{length} of a clause judgement proof
is the number of
judgements that it contains.
In a clause judgement proof, we refer to judgements
that are not derived by the rules (upward flow) and
(downward flow) as \emph{non-flow judgements}.

Let us emphasize that we allow resolution 
over both existential and universal variables,
and the resolvent must be non-tautological,
because it must be a clause (for our definition of \emph{clause}).

\subsection{Relationship to the QCSP proof system}

We now define the notion of a \emph{QCSP translation}
of a QCBF instance $\psi$, which intuitively is a QCSP instance
that behaves just like $\psi$.  
When discussing QCSP translations, we will be concerned with 
structures $\relb$ that have
just one sort $s$ with $B_s = \{ 0, 1 \}$;
we slightly abuse notation and simply write $B = \{ 0, 1 \}$.

\begin{definition} 
\normalfont
When $\psi$ is a QCBF instance, define a \emph{QCSP translation} of $\psi$
to be a QCSP instance $(\phi, \relb)$ where 
$\relb$ is a one-sorted structure with
$B = \{ 0, 1 \}$
and where 
$\phi$ is obtainable from $\psi$ by 
replacing each clause $\gamma$ having variables $v_1, \ldots, v_k$
with an atom $R(v_1, \ldots, v_k)$
such that 
$$R^{\relb} = \{ (f(v_1), \ldots, f(v_k)) ~|~ 
\mbox{$f: \{ v_1, \ldots, v_k \} \to \{ 0, 1 \}$ satisfies $\gamma$} \};$$
we typically assume that $I_{\phi} = I_{\psi}$ and that
each subformula of $\phi$ has the same index as the natural
corresponding
subformula of $\psi$.
\end{definition}

Note that when $\psi$ is a QCBF instance and $(\phi, \relb)$ is a QCSP
translation thereof,
it can be immediately verified, by induction, that for each index $i$,
an assignment $g$ to $\{ 0, 1 \}$ that is defined on
$\free(\psi(i)) = \free(\phi(i))$
satisfies $\psi(i)$ if and only if it satisfies $\phi(i)$.
In particular, we have that $\psi$ is true if and only if $\phi$ is
true on $\relb$.

We prove that our clause judgement proof system
is a faithful interpretation of our QCSP proof system,
as made precise by the following theorem.

\begin{theorem}
\label{thm:relationship}
Let $\psi$ be a QCBF instance and let $(\phi, \relb)$
be a QCSP translation of $\psi$.
For each clause judgement $(i, \alpha)$
that is derivable from $\psi$,
there exists a constraint judgement $(i, \vars(\alpha), F)$
derivable from $(\phi, \relb)$
such that the unique $g: \vars(\alpha) \to \{ 0, 1 \}$
that does not satisfy $\alpha$ is not in $F$.
The other way around,
for each constraint judgement $(i, V, F)$ that is derivable from
$(\phi, \relb)$,
and for each mapping $g: V \to \{ 0, 1 \}$ with $g \notin F$,
there exists a clause judgement $(i, \alpha)$
derivable from $\psi$ where $\vars(\alpha) \subseteq V$
and $g$ does not satisfy $\alpha$.
Consequently, 
an empty clause judgement is derivable from $\psi$
if and only if
an empty constraint judgement is derivable from $(\phi,
\relb)$.
\end{theorem}

The proof of this theorem is provided in Section~\ref{sect:thm:relationship}.

\subsection{Simulation of Q-resolution}

We now show that our clause judgement proof system
simulates Q-resolution~\cite{BuningKarpinskiFlogel95-resolution}, 
as made precise by the following theorem.

\begin{theorem}
\label{thm:simulation-of-qres}
Let $\psi$ be a QCBF instance in prenex form,
whose quantifier-free part is a conjunction of clauses
with index $c$.
If a clause $\gamma$ is derivable from $\psi$ by Q-resolution,
then the clause judgement $(c, \gamma)$
is derivable from $\psi$ by the clause judgement proof system.
\end{theorem}

\begin{proof}
It is straightforwardly verified that each clause derivable
by Q-resolution from $\psi$
is contained in the smallest set $\fancyc$
of clauses satisfying the following recursive definition:
\begin{itemize}

\item Each clause $\alpha$ appearing in $\psi$ is in $\fancyc$.

\item $\fancyc$ is closed under taking resolvents.

\item If $\alpha \in \fancyc$ and $y \in \vars(\alpha)$
is universally quantified and is the first variable in $\vars(\alpha)$
to be quantified on the unique path from $c$ to the root of $\psi$,
then $\alpha \setminus \{ y, \overnot{y} \}$ is in $\fancyc$.

\end{itemize}

\noindent It suffices to show, then, that for each $\alpha \in \fancyc$,
the clause judgement $(c, \alpha)$ is derivable.
We consider the three types of clauses according to the just-given
recursive definition.
For each clause $\alpha$ appearing in $\psi$,
the clause judgement $(c, \alpha)$ is derivable by 
applying the (clause) rule at the location of $\alpha$,
followed by one application of the (upward flow) rule.
For a clause that is a resolvent of two other clauses,
one can simply apply the (resolve) rule.
Finally, suppose that $(c, \alpha)$ is derivable and that 
$y \in \vars(\alpha)$ satisfies the described condition.
We need to show that $\alpha \setminus \{ y, \overnot{y} \}$
is derivable.
Let $j$ be the first location where $y$ is quantified
when walking from $c$ to the root, and
let $K$ be the set of nodes appearing on the unique path from $c$ 
to the child of $j$ (inclusive).
By the definition of $\fancyc$, 
no variable in $\vars(\alpha)$
is quantified at a location in $K$ and $\vars(\alpha) \subseteq
\free(\psi(k))$
for each $k \in K$; 
hence, (upward flow) can be applied repeatedly to derive
$(k, \alpha)$ for each $k \in K$.
By applying ($\forall$-removal) at the child of $j$,
we obtain $(j, \alpha \setminus \{ y, \overnot{y} \}$;
then, (downward flow) can be applied repeatedly to derive
$(c, \alpha \setminus \{ y, \overnot{y} \}$.
\end{proof}

The soundness of Q-resolution
(derivability of an empty clause implies falsehood)
is thus a consequence of this theorem,
Theorem~\ref{thm:relationship},
and Theorem~\ref{thm:sound-and-complete}.

\subsection{Algorithmic interpretation}
\label{subsect:algorithm}

We define the following notions relative to 
a QCBF instance $\psi$.
We view $\psi$ as a rooted tree.
When $i, j \in I_{\psi}$, we write 
$i \leq_{\psi} j$ if $i$ is an ancestor of $j$, that is, if $i$ occurs
on the unique path from $j$ to the root;
we write $i <_{\psi} j$ if $i \leq_{\psi} j$ and $i \neq j$.
We define a \emph{located variable} to be a pair $(i, u)$ where
$i \in I_{\psi}$ is an index, and $u$ is a variable that is quantified
at location $i$;
this pair is a \emph{$\forall$-located variable} 
if $u$ is universally quantified
at location $i$.
We say that index $j \in I_{\psi}$ \emph{follows} a located variable
$(i, u)$ 
if $i <_{\psi} j$ 
and for each index $k \in I_{\psi}$ with
$i <_{\psi} k \leq_{\psi} j$, it holds that $u \in \free(\phi(k))$.
We say that a located variable $(j, v)$ \emph{follows} a located
variable
$(i, u)$ if $j$ follows $(i, u)$.
We say that an index $j$ 
or a located variable $(j, v)$
\emph{follows} a set $S$ of located variables
if $j$ 
follows each located variable in $S$.
A set $S$ of located variables is \emph{coherent} 
if for any two distinct elements $(i, u), (j, v) \in S$,
one follows the other 
(that is, either $(i, u)$ follows $(j, v)$
or $(j, v)$ follows $(i, u)$).
When $S$ is a set of located variables, we use $\vars(S)$
to denote the set of variables occurring in 
the located variables in $S$.
Observe that when a set $S$ of located variables is coherent,
no variable occurs in two distinct located variables in $S$, and so
$|S| = |\vars(S)|$.

We now present a nondeterministic, recursive algorithm 
that, in a sense to be made precise, corresponds to the proof system.
At each point in time, the algorithm maintains a set $S$
of coherent variables;
actions it may perform include branching on 
a located variable $(i, u)$ such that adding $(i, u)$ to $S$
is still coherent, and, nondeterministically setting
a $\forall$-located variable $(i, y)$ that follows $S$.
The algorithm returns either the false value $F$ or the indeterminate value
$\bot$.  On these two values, we define the operation $\vee$ by
$F \vee F = F$ and $\bot \vee F = F \vee \bot = \bot \vee \bot =
\bot$.  Intuitively speaking, the algorithm 
returns the indeterminate value $\bot$ when
a nondeterministically selected action cannot be carried out.
We assume that when the algorithm is first invoked on a given QCBF
instance,
the set $S$ is initially assigned to the empty set.

{\footnotesize{ 
\begin{alltt}
Algorithm Detect_Falsity( QCBF instance \(\psi\), coherent set \(S\),
                            assignment \(a:\vars(S)\to\{0,1\}\))
\{
  Select nondeterministically and perform one of the following:

  (falsify)  check if there exists a location \(i\) following \(S\)
  such that \(\psi(i)\) is a clause falsified by \(a\) with \(\vars(\psi(i))=\vars(S)\);
  if so, return F, else return \(\bot\);

  (Q-branch)  check if there exists a located variable \((i,u)\notin\hspace{.1pt}S\)
  such that \(S\cup\{(i,u)\}\) is coherent; if not, return \(\bot\), else:
    - nondeterministically select such a located variable \((i,u)\);
    - nondeterministically pick subsets \(S\sb{0},S\sb{1}\subseteq\hspace{.1pt}S\) with \(S\sb{0}\cup\hspace{.1pt}S\sb{1}=S\);
    - return Detect-Falsity(\(\psi,S\sb{0}\cup\{(i,u)\},(a\res\vars(S\sb{0}))[u\to0]\)) \(\vee\)
             Detect-Falsity(\(\psi,S\sb{1}\cup\{(i,u)\},(a\res\vars(S\sb{1}))[u\to1]\))

  (\(\forall\)-branch)  check if there exists a \(\forall\)-located variable \((i,y)\) 
  that follows \(S\); if not, return \(\bot\), else:
    - nondeterministically select such a \(\forall\)-located variable \((i,y)\);
    - nondeterministically pick a value \(b\in\{0,1\}\);
    - return Detect-Falsity(\(\psi,S\cup\{(i,y)\},a[y\to\vspace{.1pt}b]\))
\}
\end{alltt}
}}

Relative to a clause judgement proof, we employ the 
following terminology.
A clause judgement $(j, \beta)$ that is derived using a previous
judgement $(i, \alpha)$ is said to be a \emph{successor} of 
$(i, \alpha)$;
also, $(i, \alpha)$ is said to be a \emph{predecessor} of 
$(j, \beta)$.
So, a clause judgement derived using the (clause) rule has $0$
predecessors,
one derived using the (resolve) rule has $2$ predecessors,
and one derived using one of the other rules has $1$ predecessor.
We say that a clause judgement proof is \emph{tree-like}
if each clause judgement has at most one successor;
in this case, each clause judgement $(i, \alpha)$
naturally induces a tree (defined recursively) where:
each node is labelled with a clause
judgement;
the root is labelled with $(i, \alpha)$;
and, for each predecessor of the clause judgement $(i, \alpha)$,
the root has a child which is the tree of the predecessor.

We formalize the notion of a trace of the nondeterministic algorithm.
A \emph{trace} of a QCBF instance $\psi$ is a rooted tree where:
\begin{itemize}

\item  each node has a label
$(S, a)$ where $S$ is a coherent set of located variables and
$a: \vars(S) \to \{ 0, 1 \}$ is an assignment; 

\item each node has $0$, $1$, or $2$ children;

\item when a node has $2$ children and label $(S, a)$, 
the labels of the two children
could be generated by the (Q-branch) step from $(S, a)$;

\item when a node has $1$ child and label $(S, a)$,
the label of the child
could be generated by the ($\forall$-branch) step from $(S, a)$;

\item when a node has $0$ children and label $(S, a)$,
the node has an associated index $i$ that follows $S$
and such that $\psi(i)$ is a clause falsified by $a$
with $\vars(\psi(i)) = \vars(S)$.

\end{itemize}

\noindent We take it as evident that this notion of trace
properly formalizes the recursion trees that the algorithm generates.

Let $e$ denote the unique assignment from $\emptyset$ to $\{ 0, 1 \}$.
We now show that, up to polynomial-time computable translations,
tree-like clause judgement proofs of an empty clause correspond
precisely to traces having root label $(\emptyset, e)$.

\begin{theorem}
\label{thm:trace-proof}
Let $\psi$ be a QCBF instance; let $n \geq 1$.
There exists a tree-like clause judgement proof $P$
(viewed as a tree) of an empty clause with $n$ non-flow judgements
if and only if
there exists a trace $T$ whose root has label $(\emptyset, e)$
and having $n$ nodes.
Moreover, both implied translations 
(from proof to trace, and from trace to proof)
can be computed in polynomial time.
\end{theorem}

The proof of this theorem is provided
in Section~\ref{sect:thm:trace-proof}.

\section{Algebraic characterization of $k$-judge-consistency}
\label{sect:consistency}

\emph{We will assume that all structures under discussion in this section
are finite, in that each structure's universe is finite.}

\begin{definition}
\normalfont
Let $k \geq 1$.
A QCSP instance $(\phi, \relb)$ is \emph{$k$-judge-consistent}
if there does not exist a judgement proof 
of width less than or equal to $k$
that contains an empty judgement.
\end{definition}

\begin{definition} \label{def:constraint-system}
\normalfont
Let $(\phi, \relb)$ be a QCSP instance, where $\phi$ is a qc-formula
with index set $I$, and let $k \geq 1$.
A \emph{$k$-constraint system} $P$ provides,
for each $i \in I$ and 
each $V \subseteq \free(\phi(i))$ with $|V| \leq k$,
a non-empty set $P[i, V]$ of maps from $V$ to $B$ 
satisfying the following four properties:
\begin{itemize}

\item $(\alpha)$ 
If
$\phi(i)$ is an atom $R(v_1, \ldots, v_m)$ with 
$V = \{ v_1, \ldots, v_m \}$,
then \\
$P[i, \{ v_1, \ldots, v_m \}] \subseteq \{ f: \{ v_1, \ldots, v_m \}
\to B ~|~ (f(v_1), \ldots, f(v_m)) \in R^{\relb} \}$


\item $(\pi)$ If $U \subseteq V$, then
$P[i, U] = (P[i, V] \res U)$.

\item $(\lambda)$ If $j$ is a child of $i$
and $V \subseteq \free(\phi(j))$, 
then
$P[i, V] = P[j, V]$.

\item $(\epsilon)$ If $j$ is a child of $i$,
$\phi(i) = \forall y \phi(j)$, 
$U$ is a subset of $\free(\phi(j))$ with $|U| \leq k$ and $y \in U$,
and $V = U \setminus \{ y \}$,
then $P[i, V] \subseteq \epsilon_y( P[j, U] )$.

\end{itemize}
\end{definition}

\noindent We show that the existence of a $k$-constraint system
characterizes $k$-judge consistency.

\begin{theorem}
\label{thm:consistency-characterization}
Let $(\phi, \relb)$ be a QCSP instance.
There exists a $k$-constraint system $P$
for the instance if and only if the instance is $k$-judge-consistent.
\end{theorem}

A proof of this theorem can be found in
Section~\ref{sect:thm:consistency-characterization}.

\begin{theorem}
\label{thm:deciding-consistency-ptime}
For each $k \geq 1$,
there exists a polynomial-time algorithm that,
given a QCSP instance $(\phi, \relb)$,
decides if the instance is $k$-judge-consistent.
\end{theorem}

\begin{proof}
We begin by describing the algorithm, which
decides, given a QCSP instance $(\phi, \relb)$,
 whether or not there exists a $k$-constraint system
(this property is equivalent to $k$-judge-consistency
by Theorem~\ref{thm:consistency-characterization}).
Throughout, $i$ and $j$ will always denote
indices from $I_{\phi}$.

For each $i \in I_{\phi}$ and $V \subseteq \free(\phi(i))$
with $|V| \leq k$,
the algorithm initializes
$Q[i, V]$ to be 
$\{ f: \{ v_1, \ldots, v_m \}
\to B ~|~ (f(v_1), \ldots, f(v_m)) \in R^{\relb} \}$
in the case that $\phi(i)$ is an atom $R(v_1, \ldots, v_m)$
and $V = \{ v_1, \ldots, v_m \}$,
and otherwise 
initializes $Q[i, V]$
to be the set of all maps from $V$ to $B$.
The algorithm then iteratively performs the following rules
(which parallel properties $(\pi)$, $(\lambda)$, and $(\epsilon)$
in the definition of $k$-constraint system)
until no changes can be made to $Q$:
\begin{itemize}

\item When $i$ is an index
 and $U \subseteq V \subseteq \free(\phi(i))$
with $|V| \leq k$, 
assign to $Q[i, U]$ the value
$Q[i, U] \cap (Q[i, V] \res U)$
and then
assign to $Q[i, V]$ the value
$\{ f \in Q[i, V] ~|~ (f \res U) \in Q[i, U] \}$.

\item If $j$ is a child of $i$,
$V \subseteq \free(\phi(i)) \cap \free(\phi(j))$,
and $|V| \leq k$,
assign to each of $Q[i, V]$, $Q[j, V]$
the value $Q[i, V] \cap Q[j, V]$.

\item If $j$ is a child of $i$,
$\phi(i) = \forall y \phi(j)$,
and $U$ is a set of variables with
$y \in U \subseteq \free(\phi(j))$ and $|U| \leq k$,
then assign to
$Q[i, U \setminus \{ y \}]$
the value $Q[i, U \setminus \{ y \}] \cap \epsilon_y(P[j, U])$.

\end{itemize}

\noindent This algorithm runs in polynomial time: 
there are polynomially many pairs $(i, V)$ for which
$Q[i, V]$ is initialized and used, and
each $Q[i, V]$ contains (at most) polynomially many maps:
when $|V| \leq k$, the number of maps from $V$ to $B$
is polynomial.  (Here, when we say \emph{polynomial},
we mean as a function of the input length.)
Applying the three given rules can be done in polynomial time;
each time they are applied, the sets $Q[i, V]$
may only decrease in size.
Hence, the process of repeatedly applying the three rules
until no changes are possible terminates in polynomial time.

We now explain why the instance is $k$-judge-consistent
if and only if no set $Q[i, V]$ is empty, which suffices
to give the theorem.
It is straightforward to verify that,
for any $k$-constraint-system $P$,
the invariant $P[i, V] \subseteq Q[i, V]$
is maintained by the algorithm.
Hence, when the algorithm terminates,
if any set $Q[i, V]$ is empty, then there does not exist
a $k$-constraint system $P$.
It is also straightforward to verify that,
when the algorithm terminates,
the four properties in the definition of $k$-constraint system
hold on $Q$.
(As an example, consider property ($\lambda$).
Suppose that $j$ is a child of $i$ and that
$V \subseteq \free(\phi(i)) \cap \free(\phi(j))$.
When the algorithm terminates,
since the second rule can no longer be applied
it must hold that
$Q[i,V] = Q[j,V] = Q[i,V] \cap Q[j,V]$.)
Hence, if the algorithm terminates without any 
empty set $Q[i, V]$, it holds that $Q$ is a $k$-constraint system.
\end{proof}

We can upper bound the number of iterations that the algorithm
performs on an instance $(\phi, \relb)$ in the following way.
Let $n$ be the maximum number of free variables, over
all subformulas of $\phi$.
For each index $i$ of $\phi$, the algorithm maintains,
for each $V \subseteq \free(\phi(i))$ with $|V| \leq k$,
a set of mappings from $V$ to $B$.
The size of such a set is at most $|B|^{|V|}$.
In each iteration, each such set of mappings
can only have mappings deleted from it.
The number of iterations is thus upper bounded by
the number of mappings that can occur in such sets of mappings,
which is
$|I_{\phi}|( {n \choose k} |B|^k + {n \choose k-1}|B|^{k-1} + \cdots + 
{n \choose 0} |B|^0 )$.

We now show that checking for $k$-judge-consistency gives a way to decide
a set of prenex qc-sentences 
that is tractable via the dichotomy theorem
on so-called prefixed graphs~\cite{ChenDalmau12-decomposingquantified}.
In particular, we prove this in the setting 
where relation symbols have bounded arity.
Let us refer to the width notion defined 
in that previous work~\cite{ChenDalmau12-decomposingquantified}
as \emph{elimination width}.
Define the \emph{Q-width} of a prenex qc-sentence $\phi$ to be
the maximum of its elimination width
and $\max_R |\ar(R)|$ (where this maximum ranges over
all relation symbols $R$ appearing in $\phi$).

\begin{theorem}
\label{thm:consistency-solves-qwidth}
Let $k \geq 1$.  Suppose that $\phi$ is a prenex qc-sentence with Q-width
$k$ (or less).
For any finite structure $\relb$, it holds that
$(\phi, \relb)$ is $k$-judge-consistent if and only if $\relb \models \phi$.
(Intuitively, this says that 
checking for $k$-judge consistency is a decision procedure
for QCSP instances involving $\phi$.)
\end{theorem}

This theorem, 
in conjunction 
with Theorem~\ref{thm:deciding-consistency-ptime},
immediately implies that for any set $\Phi$ of qc-sentences
having Q-width bounded by a constant $k$,
checking for $k$-judge-consistency 
is a uniform polynomial-time procedure
that decides any QCSP instance $(\phi, \relb)$
where $\phi \in \Phi$ and $\relb$ is finite.
Hence, in the setting of bounded arity, checking for $k$-judge-consistency
is a \emph{generic} reasoning procedure 
that correctly decides the tractable cases of QCSP
identified by the work on elimination width.

In order to establish this theorem, we first prove a lemma.

\begin{lemma}
\label{lemma:two-transformations}
Suppose that the QCSP instance $(\theta, \relb)$ is $k$-judge-consistent,
that $\relb$ is a finite structure,
and that $\theta'$ is a qc-sentence obtained from $\theta$
by applying one of the following three syntactic transformations
to a subformula of $\theta$:
\begin{enumerate}

\item $\bigwedge_{i \in I} \phi_i \leadsto 
(\bigwedge_{j \in J} \phi_j) \wedge (\bigwedge_{k \in K} \phi_k)$,
where $I$ is the disjoint union of $J$ and $K$

\item $Qv (\phi \wedge \psi) \leadsto (Qv \phi) \wedge \psi$
where $v \notin \free(\psi)$

\item $\forall y \bigwedge_{i \in I} \phi_i \leadsto
\bigwedge_{i \in I} (\forall y \phi_i)$

\end{enumerate}
Then, the QCSP instance $(\theta', \relb)$ is $k$-judge-consistent.
\end{lemma}

\begin{proof}
By Theorem~\ref{thm:consistency-characterization},
it suffices to show that if $(\theta, \relb)$
has a $k$-constraint system $P$,
then $(\theta', \relb)$
does as well.
We consider each of the three cases.

Case (1):
We define 
a $k$-constraint system $P'$ for $(\theta', \relb)$
in the following way.
Relative to the transformation,
let $i$ denote the index of $\phi_i$
in both $\theta$ and $\theta'$;
let $c$ denote the index of 
$\bigwedge_{i \in I} \phi_i$ in $\theta$
and of 
$(\bigwedge_{j \in J} \phi_j) \wedge (\bigwedge_{k \in K} \phi_k)$
in $\theta'$;
let $a$ be the index of 
$\bigwedge_{j \in J} \phi_j$ in $\theta'$;
and let $b$ be the index of 
$\bigwedge_{k \in K} \phi_k$ in $\theta'$.
For each other subformula occurrence in $\theta'$, 
there is a corresponding subformula occurrence
in $\theta$; we will assume that these two corresponding
subformula occurrences share the same index.

We now describe how to define $P'$.
Whenever discussing $P'[d, V]$,
it will hold that
$d$ is an index of $\theta'$, and we assume that
$V \subseteq \free(\theta'(d))$ and $|V| \leq k$.
We define $P'[i, V]$ as $P[i, V]$.
We define $P'[c, V]$ as $P[c, V]$.
We define 
$P'[a, V]$ as $P[c, V]$,
and similarly we define
$P'[b, V]$ as $P[c, V]$.
For each other index $\ell$ of $\theta'$,
we define $P'[\ell, V]$ as $P[\ell, V]$.
It is straightforward to verify that $P'$
is a $k$-constraint system.

Case (2):
We proceed as in the previous case; 
we define a $k$-constraint system $P'$ for $(\theta', \relb)$.
Relative to the transformation, let $a$ denote the index of
$Qv (\phi \wedge \psi)$ in $\theta$;
let $b$ denote the index of the subformula
$Qv \phi$ in $\theta'$.
We define $P'[b, V]$ as $P[a, V]$.
For each other subformula occurrence of $\theta'$
with index $\ell$,
there exists a corresponding subformula occurrence
of $\theta$ which we assume to also have index $\ell$.
We define $P'[\ell, V]$ as $P[\ell, V]$.
It is straightforward to verify that $P'$
is a $k$-constraint system.

Case (3):
We proceed as in the previous cases;
we define a $k$-constraint system $P'$ for $(\theta', \relb)$.
Let $d$ denote the index of 
$\forall y \bigwedge_{i \in I} \phi_i$ in $\theta$,
and also the index of
$\bigwedge_{i \in I} (\forall y \phi_i)$ in $\theta'$.
Let $i$ denote the index of $\phi_i$ in both $\theta$ and $\theta'$.
Let $c$ denote the index of $\bigwedge_{i \in I} \phi_i$ in $\theta$,
and let $i'$ denote the index of $\forall y \phi_i$ in $\theta'$.
For each $V \subseteq \free(\forall y \phi_i)$ with $|V| \leq k$,
define $P'[i', V]$ to be $P[d, V]$.
Elsewhere, define $P'$ to be equal to $P$ (each other index of
$\theta'$
corresponds to an index of $\theta$).
It is straightforward to verify that $P'$ is a $k$-constraint system.
In the region of interest, the property $(\epsilon)$
can be verified as follows.
Suppose that $U \subseteq \free(\phi(i))$ has $|U| \leq k$ and $y \in
U$,
and that $V = U \setminus \{ y \}$.
Then $P[d, V] \subseteq \epsilon_y(P[c, U]) = \epsilon_y(P[i, U])$
since $P$ is a $k$-constraint system.
As $P'[i', V] = P[d, V]$ by our definition of $P'$,
it follows that $P'[i', V] \subseteq \epsilon_y(P[i, U])$.
\end{proof}

\begin{proof} (Theorem~\ref{thm:consistency-solves-qwidth})
Suppose that the instance $(\phi, \relb)$ is not $k$-judge-consistent.
Then, by definition, there exists a judgement proof for the instance
containing an empty judgement, implying that $\relb \not\models \phi$
by Theorem~\ref{thm:sound-and-complete}.

For the other direction, suppose that $\relb \not\models \phi$.
From the definition of elimination width (defined as \emph{width}
in~\cite{ChenDalmau12-decomposingquantified}),
it can straightforwardly be verified by induction on the number of
variables in $\phi$ that $\phi$ can be transformed to a
sentence $\phi'$ having width less than or equal to $k$, via
the three syntactic transformations
of Lemma~\ref{lemma:two-transformations}.
As these three syntactic transformations preserve logical equivalence,
we have $\relb \not\models \phi'$.
By Theorem~\ref{thm:sound-and-complete},
an empty judgement is derivable;
by the proof of this theorem, there is a judgement proof 
with the empty judgement whose width
is equal to the
width of $\phi'$.  Since the width of $\phi'$ is less than or equal to
$k$,
we thus obtain a judgement proof of the empty judgement 
having width less than or equal to $k$, so by definition,
$(\phi', \relb)$ is not $k$-judge-consistent.
By appeal to Lemma~\ref{lemma:two-transformations},
$(\phi, \relb)$ is not $k$-judge-consistent.
\end{proof}

\section*{Acknowledgements} 
The author thanks Moritz M\"{u}ller and Friedrich Slivovsky 
for useful comments.
This work was supported by the Spanish project
TIN2013-46181-C2-2-R, 
by the Basque project GIU12/26,
and by the Basque grant UFI11/45.

\newpage

\appendix

\section{Proof of Theorem~\ref{thm:relationship}}
\label{sect:thm:relationship}

The theorem follows directly from the following two theorems.

\begin{theorem}
\normalfont
Let $\psi$ be a QCBF instance and let $(\phi, \relb)$ be a QCSP translation
of $\psi$.
For each clause judgement proof of $\psi$
having length $s$ and width $w$,
there exists a constraint judgement proof of $(\phi, \relb)$
having length $\leq 2s$ and width $\leq w+1$
such that: each clause judgement $(i, \alpha)$ appearing in the
first proof has the entailment property that
there exists a constraint judgement in the second proof
of the form $(i, \vars(\alpha), F)$ 
such that
each $f \in F$ satisfies $\alpha$
(equivalently, the unique $g: \vars(\alpha) \to \{ 0, 1 \}$
that does not satisfy $\alpha$ is not in $F$).
\end{theorem}

A direct consequence of this theorem is that
if the original clause judgement proof contains an empty clause,
then the produced constraint judgement proof contains an empty constraint.

\begin{proof}
We prove this by induction on $s$.
Given a clause judgement proof $P$ of length $s + 1$
we create a constraint judgement proof in the following way.
Apply induction to the clause judgement proof consisting
of the first $s$ judgements in $P$; this gives a constraint judgement
proof $P'$.
We then need to show how to augment $P'$.
We consider cases, depending on the rule used to derive the
last judgement of $P$.  We use the notation of
Definition~\ref{def:clause-judgement-proof}.

\begin{itemize}
\item In the case of (clause) deriving $(i, \alpha)$,
apply (atom) at location $i$.

\item In the case of (resolve) deriving $(i, \gamma)$ from
$(i, \alpha)$ and $(i, \beta)$, let $v$ be the variable underlying
the complementary literals that are eliminated from $\alpha$ and
$\beta$
to obtain $\gamma$.  
The rule (join) is applied to the constraint judgements corresponding
to $(i, \alpha)$ and $(i, \beta)$ to obtain a new
judgement,
and then (projection) is used to remove the variable $v$ from that new judgement.

\item In the case of (upward flow) or (downward flow), the same rule is
applied
to the corresponding constraint judgement.

\item In the case of ($\forall$-removal), the rule ($\forall$-elimination)
is applied to the corresponding constraint judgement.

\end{itemize}
In the case (resolve), two new constraint judgements are produced,
and in all other cases, one new constraint judgement is produced;
hence, the claim on the length is correct.
In the case (resolve), the width of the first constraint judgement
produced
is one more than the width of the corresponding clause judgement,
and the width of the second constraint judgement produced
is equal to the width of the corresponding clause judgement;
in all other cases, the new constraint judgement produced
has width equal to that of the corresponding clause judgement.
Hence, the claim on the width is correct.
In each case, it is straightforward to verify the claimed entailment property.

As an example, we verify the claimed entailment property
in the case of (resolve).
Suppose that (resolve) derives $(i, \gamma)$ from
$(i, \alpha)$ and $(i, \beta)$.
Let 
$g_{\alpha}: \vars(\alpha) \to \{ 0, 1 \}$,
$g_{\beta}: \vars(\beta) \to \{ 0, 1 \}$, and
$g_{\gamma}: \vars(\gamma) \to \{ 0, 1 \}$
be assignments not satisfying $\alpha$, $\beta$, and $\gamma$,
respectively.
Let $v$ be the variable such that 
$\vars(\gamma) = (\vars(\alpha) \cup \vars(\beta)) \setminus \{ v \}$.
We assume without loss of generality that
$g_{\alpha}(v) = 0$ and that $g_{\beta}(v) = 1$.
Let 
$(i, \vars(\alpha), F_{\alpha})$
and
$(i, \vars(\beta), F_{\beta})$
be the constraint judgements for $(i, \alpha)$ and $(i, \beta)$,
respectively; 
we have $g_{\alpha} \notin F_{\alpha}$
and $g_{\beta} \notin F_{\beta}$.
Consider the constraint judgement
$(i, \vars(\alpha) \cup \vars(\beta), F_{\alpha} \Join F_{\beta})$
obtained by applying (join) to these two constraint judgements.
By definition of the join $\Join$, 
neither $g_{\alpha}$ not $g_{\beta}$ 
has an extension defined on $\vars(\alpha) \cup \vars(\beta)$
that is contained in $F_{\alpha} \Join F_{\beta}$.
Next, consider the constraint judgement 
$(i, \vars(\gamma), (F_{\alpha} \Join F_{\beta}) \res \vars(\gamma))$
obtained from the previous one by applying (projection).
We claim that 
$g_{\gamma} \notin (F_{\alpha} \Join F_{\beta}) \res \vars(\gamma)$.
Suppose not, for a contradiction;
then there exists an extension $g'_{\gamma}$ of $g_{\gamma}$
which is contained in $F_{\alpha} \Join F_{\beta}$.
If $g'_{\gamma}(v) = 0$, then $g'_{\gamma}$
is an extension of $g_{\alpha}$,
but $g'_{\gamma} \in F_{\alpha} \Join F_{\beta}$
contradicts $g_{\alpha} \notin F_{\alpha}$;
analogously,
If $g'_{\gamma}(v) = 1$, then $g'_{\gamma}$
is an extension of $g_{\beta}$,
but $g'_{\gamma} \in F_{\alpha} \Join F_{\beta}$
contradicts $g_{\beta} \notin F_{\beta}$.
\end{proof}

\begin{theorem}
\normalfont
Let $\psi$ be a QCBF instance and let $(\phi, \relb)$ be a QCSP translation
of $\psi$.
For each constraint judgement proof of $(\phi, \relb)$ 
having length $s$ and width $w$,
there exists a clause judgement proof of $\psi$
of length $\leq s \cdot \max(w2^{w-1}, 1)$ and width $\leq w$
such that:
each constraint judgement $(i, V, F)$ 
appearing in the first proof
has the entailment property that,
for each mapping $g: V \to \{ 0, 1 \}$ with $g \notin F$,
there exists a clause judgement $(i, \alpha)$ 
with $\vars(\alpha) \subseteq V$
in the second proof 
where $\alpha$ is not satisfied by $g$.
\end{theorem}

A direct consequence of this theorem is that
 if a constraint judgement proof of $(\phi, \relb)$ 
having length $s$ and width $w$
contains an empty
constraint,
it may be augmented by one constraint judgement to contain
an empty constraint of the form $(i, \emptyset, \emptyset)$,
and then the theorem yields that there is 
a clause judgement proof having an empty clause of
length $\leq (s+1) \cdot \max(w2^{w-1}, 1)$ and width $\leq w$.

\begin{proof}
We proceed as in the proof of the previous theorem.
We prove this by induction on $s$.
Given a constraint judgement proof $P$ of length $s+1$,
we create a clause judgement proof $P'$ by applying induction
to $P$ with the last constraint judgement removed;
we then explain how to augment the resulting clause judgement proof $P'$
so that the last constraint judgement of $P$ has a corresponding
clause
judgement with the properties given in the theorem statement.
We consider cases depending on the rule used to derive the last
constraint judgement of $P$; 
we use the notation of Definition~\ref{proof-system}.

\begin{itemize}

\item In the case of (atom) deriving $(i, V, F)$, apply
(clause) at location $i$.

\item In the case of (projection) deriving
$(i, U, F \res U)$ from $(i, V, F)$, we first explain how to obtain
the clause judgements in the case that $|V| = |U| + 1$.
Let $v$ be the variable such that $U \cup \{ v \} = V$.
For each clause judgement $(i, \alpha)$
with $\vars(\alpha) \subseteq U$ that can be obtained by 
resolving two clause judgements in $P'$ on the variable $v$,
include the clause judgement in the proof.
The maximum number of clause judgements that we can add in this
fashion
is the number of clauses on $(w-1)$ variables, that is, $2^{w-1}$.

In the general case where $U \subseteq V$, 
we may proceed by applying the described procedure
$|V| - |U|$ many times.  Since $|V| - |U| \leq w$, the total
number of clause judgements that will be added can be upper bounded
by $w 2^{w-1}$.

\item In the case of (join), no clause judgement needs to be added.
This is because of the following.
Suppose the constraint judgement
$(i, U_1 \cup U_2, F_1 \Join F_2)$ is obtained by applying
(join) to $(i, U_1, F_1)$ and $(i, U_2, F_2)$.
For each mapping $g: U_1 \cup U_2 \to \{ 0, 1 \}$
with $g \notin F_1 \Join F_2$,
it holds (by definition of $\Join$)
that either $g \res U_1 \notin F_1$ or $g \res U_2 \notin F_2$.

\item In the case of ($\forall$-elimination)
deriving $(i, V \setminus \{ y \}, \epsilon_y F)$
from $(j, V, F)$, take all clause judgements $(j, \alpha)$
where $y \in \vars(\alpha) \subseteq V$, and apply 
($\forall$-removal) to each of these clause judgements.

\item In the case of (upward flow) or (downward flow),
the same rule is applied to the corresponding clause judgement.

\end{itemize}
In each case, the clause judgements produced have width less than or
equal to $w$.
We now consider
the number of clause judgements produced in each case.
This number is $1$
in the cases (atom), (upward flow), and (downward flow),
and is $0$
in the case (join).
In the case of (projection), we argued that this number is less than
or equal to $w 2^{w-1}$.
In the case of ($\forall$-elimination), since this rule can only be
applied if
$w \geq 1$ and at most $2^{w-1}$ clauses are generated, we can
also bound this number by $w 2^{w-1}$.

In each case, it is straightforward to verify the claimed entailment property.
\end{proof}

\section{Proof of Theorem~\ref{thm:trace-proof}}
\label{sect:thm:trace-proof}
The theorem follows directly from the following two theorems.

\begin{theorem}
Let $\psi$ be a QCBF instance.
Given a tree-like clause judgement proof $P$ (viewed as a tree)
of an empty clause,
there exists a trace $E$ whose root has label 
$(\emptyset, e)$ and where 
the number of nodes in $E$ 
is equal to the number of non-flow judgements in $P$.
(Here, we use $e$ to denote the unique assignment from $\emptyset$
to $\{ 0, 1 \}$.)
Also, the translation from $P$ to $E$ is polynomial-time computable.
\end{theorem}

\begin{proof}
We prove the following result, which yields the theorem.
Suppose that $P$ is a tree-like clause judgement proof,
viewed as a tree; 
using $(i, \alpha)$ to denote the clause judgement at the root of $P$,
there exists a trace $E$ whose number of nodes
is equal to the number of non-flow judgements in $P$,
and whose root has label $(S, a)$, such that the following two
conditions hold:
\begin{enumerate}

\item For each $v \in \vars(\alpha)$,
the set $S$ contains the located variable $(j, v)$
where $j$ is the first location above $i$ where $v$ is quantified.

\item The assignment $a$ is the unique assignment on
  $\vars(\alpha)$
that falsifies $\alpha$.

\end{enumerate}
We prove this result by induction on the structure of $P$,
describing directly how to construct $E$.

We consider cases depending on how the clause judgement at the root of
$P$
was derived; we use the notation of
Definition~\ref{def:clause-judgement-proof}.

In the case of (clause), let $E$ consist of a single node
having label $(S, a)$, where $(S, a)$ is the unique pair satisfying
the two conditions.

In the case of (resolve), suppose that $(i, \gamma)$ is the clause
judgement
at the root of $P$ and that $(i, \gamma)$
is derived as a resolvent of $\alpha$ and $\beta$
via clause judgements $(i, \alpha)$ and $(i, \beta)$.
Suppose that $v \in \alpha$ and $\overnot{v} \in \beta$
are the complementary literals such that
$\gamma = (\alpha \setminus \{ v \}) \cup (\beta \setminus \{
\overnot{v} \})$.
Take the trace whose root has label $(U, c)$
where $U$ is the union of $S$ and $T$ but without the located variable
containing $v$, and where $c$
is the unique assignment on $\vars(U) = \vars(\gamma)$
that falsifies $\gamma$.
Since $i$ follows $S$ and $i$ follows $T$,
we have that $i$ follows $U$,
and we have that
$(S, a)$ and $(T, b)$ could be generated from $(U, c)$
via a (Q-branch) step.

In the case of ($\forall$-removal), 
suppose that the clause judgement at the root of $P$
has the form $(i, \alpha \setminus \{ y, \overnot{y} \})$
and is derived from $(j, \alpha)$
where $\phi(i) = \forall y \phi(j)$.
If $\alpha \cap \{ y, \overnot{y} \} = \emptyset$, 
then the trace $E$ can be taken to be the trace given by induction.
Otherwise, take the trace for $(j, \alpha)$ given by induction, 
and let $(T, a)$ denote its root node label.
Set $S$ to be $T$, but with the located variable for $y$ removed.
We have that $(T, a)$ could be derived from
$(S, a \res \vars(S))$ by a ($\forall$-branch) step;
hence, we may take the trace obtained from the trace for $(j, \alpha)$
by adding on the top a new root node with label
$(S, a \res \vars(S))$.

In the case of (upward flow) or (downward flow),
we simply take the trace given by induction.
This preserves condition (1): 
if $i$ is the parent of $j$ in $\phi$
and $(i, \alpha)$ and $(j, \alpha)$ are clause judgements in $P$,
then $\vars(\alpha) \subseteq \free(\psi(i)) \cap \free(\psi(j))$
and so no variable in $\vars(\alpha)$ is 
quantified at location $i$ (nor $j$).
\end{proof}

\begin{theorem}
Let $\psi$ be a QCBF instance.
Given a trace $T$ with root node label $(\emptyset, e)$,
there exists a tree-like clause judgement proof $P$
of an empty clause 
where the number of non-flow nodes in $P$ (viewed as a tree)
is equal to the number of nodes in $T$.
Also, the translation from $T$ to $P$ is polynomial-time computable.
\end{theorem}

\begin{proof}
We prove the following result which implies the theorem: 
for any trace $T$ with root node label
$(S, a)$,
there exists a tree-like clause judgement proof $P$ 
ending in $(i, \alpha)$
where the number of nodes in $P$ and $T$ are related as in the theorem
statement,
and such that the following two conditions hold:
\begin{enumerate}

\item $i$ follows $S$.

\item $\vars(S) = \vars(\alpha)$ and $a$ is the unique assignment
on $\vars(\alpha)$ that falsifies $\alpha$.

\end{enumerate}
We prove the result by induction; 
we consider cases depending on the type of the root node of $T$,
that is, depending on how many children the root node of $T$ has.

If the root node of $T$ is a leaf, the result is clear from the
definition of trace.

If the root node of $T$ has one child,
let $(S \cup \{ (j, y) \}, a[y \to b])$
be the label of the child of the root node.
By induction, 
there exists a tree-like clause judgement proof ending with
$(k, \beta)$
where $k$ follows $S \cup \{ (j, y) \}$
and $a[y \to b]$ is the unique assignment on $\vars(\beta)$
that falsifies $\beta$.
Since $(j, y)$ follows $S$ and $k$ follows $(j, y)$,
by applying (upward flow), we may obtain a tree-like clause judgement
proof ending with $(c, \beta)$ where $c$ is the child of $j$.
Then, as $\psi(j) = \forall y \psi(c)$,
we can apply ($\forall$-removal) to $(c, \beta)$
to obtain the desired clause judgement proof.

If the root node of $T$ has two children,
let
$(S_0 \cup \{(j,u)\}, (a \res \vars(S_0))[u \to a_0])$
and
$(S_1 \cup \{(j,u)\}, (a \res \vars(S_1))[u \to a_1])$
be the labels of the children;
we have $a_0, a_1 \in \{ 0, 1 \}$ and $a_0 \neq a_1$.
By induction, we have tree-like clause judgement proofs ending with
$(k_0, \beta_0)$ and $(k_1, \beta_1)$
where 
$k_0$ follows $S_0 \cup \{ (j, u) \}$ and
$(a \res \vars(S_0))[u \to a_0]$
is the unique assignment on $\vars(\beta_0)$ that falsifies
$\beta_0$;
and similarly,
$k_1$ follows $S_1 \cup \{ (j, u) \}$ and
$(a \res \vars(S_1))[u \to a_1]$
is the unique assignment on $\vars(\beta_1)$ that falsifies
$\beta_1$.
By applying (upward flow), we obtain tree-like clause judgement proofs
ending with $(\ell_0, \beta_0)$ and $(\ell_1, \beta_1)$
where $\ell_0$ is the child of the lowest location in 
$S_0 \cup \{ (j,u) \}$,
and $\ell_1$ is the child of the lowest location in 
$S_1 \cup \{ (j,u) \}$.
Let $m$ be the child of the lowest location in $S \cup \{ (j,u) \}$.
At least one of $\ell_0$, $\ell_1$ is equal to $m$
(since $S_0 \cup S_1 = S$).
If one of $\ell_0$, $\ell_1$ is not equal to $m$,
say $\ell_b$,
we may apply (downward flow) to obtain a clause judgement proof
$(m, \beta_b)$; this is because 
$\vars(\beta_b)$ is free in every location between 
$m$ and $\ell_b$ (inclusive), as $S \cup \{ (j, u) \}$ is coherent.
We hence obtain clause judgement proofs for
$(m, \beta_0)$ and for $(m, \beta_1)$.
Apply (resolve) to these to obtain the desired clause judgement proof.
\end{proof}

\section{Proof of Theorem~\ref{thm:consistency-characterization}}
\label{sect:thm:consistency-characterization}

Theorem~\ref{thm:consistency-characterization}
follows immediately from the two lemmas presented in this section.

\begin{lemma}
Let $(\phi, \relb)$ be a QCSP instance.
If there exists a $k$-constraint system $P$
for the instance, then the instance is $k$-judge-consistent.
\end{lemma}

\begin{proof}
We show, by induction on the proof structure, 
that if $(i, V, F)$ with $|V| \leq k$
is a derivable judgement,
then $P[i, V] \subseteq F$.
We consider cases based on which rule was used to derive
$(i, V, F)$.

In the case of (atom), we have $P[i, V] \subseteq F$ by 
property $(\alpha)$.

In the case of (projection), 
we suppose that 
$(i, V, F)$
 is a previous judgement
with $P[i, V] \subseteq F$,
and that the judgement of interest has the form
$(i, U, F \res U)$,
where $U \subseteq V$.
We have $P[i, U] = (P[i, V] \res U) \subseteq (F \res U)$,
where the equality holds by property $(\pi)$.

In the case of (join),
we suppose that 
$(i, U_1, F_1)$ and $(i, U_2, F_2)$
are previous judgements
with 
$P[i, U_1] \subseteq F_1$
and
$P[i, U_2] \subseteq F_2$,
and that the judgement of interest is
$(i, U_1 \cup U_2, F_1 \Join F_2)$.
By property $(\pi)$, we have that
$P[i, U_1 \cup U_2] \res U_1 = P[i, U_1]$ 
and 
$P[i, U_1 \cup U_2] \res U_2 = P[i, U_2]$.
It follows that 
$P[i, U_1 \cup U_2] \subseteq P[i, U_1] \Join P[i, U_2]
\subseteq F_1 \Join F_2$.

The cases of (upward flow) and (downward flow) follow
immediately from property $(\lambda)$.

In the case of ($\forall$-elimination),
we suppose that
$(j, V, F)$ is a previous judgement
with $P[j, V] \subseteq F$,
and that the judgement of interest is
$(i, V \setminus \{ y \}, \epsilon_y F)$
where $i$ is the parent of $j$,
and $\phi(i) = \forall y \phi(j)$.
We have $P[i, V \setminus \{ y \}] \subseteq \epsilon_y( P[j, V] )
\subseteq \epsilon_y(F)$, where the first containment holds by
property
$(\epsilon)$.
\end{proof}

\begin{definition}
Let $k \geq 1$.
A structure $\relb$ is \emph{$k$-behaved} if for each $i$ with $1 \leq
i \leq k$, there are finitely many relations of arity $i$ that are
qc-definable over $\relb$.
\end{definition}

\begin{lemma}
Let $k \geq 1$.
Let $(\phi, \relb)$ be a QCSP instance
where $\relb$ is $k$-behaved.
If the instance is $k$-judge-consistent,
then then there exists a $k$-constraint system $P$ for the instance.
\end{lemma}

\begin{proof}
Relative to a QCSP instance, we say that
a judgement is \emph{$k$-derivable}
if there exists a judgement proof 
of width less than or equal to $k$
that contains the judgement.

Let $I$ be an index set for $\phi$.
Let us say that a $k$-derivable judgement 
$(i, V, F)$
is \emph{minimal} if for all sets $G$ such that
the judgement $(i, V, G)$
is $k$-derivable, it holds that
$G \subseteq F$ implies $G = F$.
We claim that,
when $i \in I$ and $V \subseteq \free(\phi(i))$ with $|V| \leq k$,
there is a unique minimal $k$-derivable judgement
$(i, V, F)$.
The existence of a minimal $k$-derivable judgement
follows from 
Lemma~\ref{lemma:judgement-gives-formula}
and
the $k$-behavedness of $\relb$.
To establish uniqueness, suppose for a contradiction that
$(i, V, F_1)$ and $(i, V, F_2)$
are both minimal $k$-derivable judgements and $F_1 \neq F_2$.
By the definition of minimal, we have
$F_1 \not\subseteq F_2$ and
$F_2 \not\subseteq F_1$,
so
$F_1 \cup F_2 \not\subseteq F_1$
and
$F_1 \cup F_2 \not\subseteq F_2$.
By the (join) rule,
the judgement 
$(i, V, F_1 \Join F_2)$
is $k$-derivable; since here $F_1 \Join F_2 = F_1 \cap F_2$, we obtain a
contradiction.

For all $i \in I$ and $V \subseteq \free(\phi(i))$,
we define $P[i, V]$ so that
$(i, V, P[i, V])$ is the 
unique minimal $k$-derivable judgement
involving $i$ and $V$.
We confirm that $P$ is a $k$-constraint system 
by verifying that it satisfies each of the four properties
of the definition of $k$-constraint system.
In discussing each of the properties,
we use the notation of Definition~\ref{def:constraint-system}.

Property $(\alpha)$ follows immediately from the (atom) rule.

For property $(\pi)$, suppose that $U \subseteq V$.
We have that
$(i, U, P[i, U])$ and
$(i, V, P[i, V])$ 
are $k$-derivable.
It follows that $(i, V, F_V)$
and
$(i, U, F_U)$
are $k$-derivable,
where $F_V = P[i, U] \Join P[i, V]$
and $F_U = F_V \res U$.
We have $F_V \subseteq P[i, V]$
and $F_U \subseteq P[i, U]$;
it follows, by definition of $P$,
that $F_V = P[i, V]$ and $F_U = P[i, U]$.
Since $F_U = F_V \res U$, we have $P[i, U] = P[i, V] \res U$.

For property $(\lambda)$,
suppose that $j$ is a child of $i$ with $V \subseteq \free(\phi(j))$.
We have that
$(i, V, P[i, V])$ 
and
$(j, V, P[j, V])$ 
are $k$-derivable.
By the (downward flow) and (upward flow) rules, we obtain
that
$(i, V, P[j, V])$ 
and
$(j, V, P[i, V])$
are $k$-derivable.
By definition of $P$, we obtain that
$P[i, V] \subseteq P[j, V]$
and
$P[j, V] \subseteq P[i, V]$
and hence $P[i, V] = P[j, V]$.

For property $(\epsilon)$,
suppose that
$j$ is a child of $i$,
$\phi(i) = \forall y \phi(j)$, 
$U$ is a subset of $\free(\phi(j))$ with $|U| \leq k$ and $y \in U$,
and $V = U \setminus \{ y \}$.
That
$P[i, V] \subseteq \epsilon_y( P[j, U] )$
follows immediately from applying the ($\forall$-elimination)
rule to the $k$-derivable judgement
$(j, U, P[j, U])$.
\end{proof}

\newpage

\bibliographystyle{abbrv}

\bibliography{../../hubiebib}

\begin{thebibliography}{10}

\bibitem{AtseriasKolaitisVardi04-propagation-as-proof-system}
A.~Atserias, P.~Kolaitis, and M.~Vardi.
\newblock Constraint propagation as a proof system.
\newblock In {\em Proceedings of CP}, 2004.

\bibitem{BartoKozik09-boundedwidth}
L.~Barto and M.~Kozik.
\newblock Constraint satisfaction problems of bounded width.
\newblock In {\em Proceedings of FOCS'09}, 2009.

\bibitem{BeamePitassi98-survey}
P.~Beame and T.~Pitassi.
\newblock Propositional proof complexity: Past, present and future.
\newblock {\em Bulletin of the EATCS}, 65:66--89, 1998.

\bibitem{BuningKarpinskiFlogel95-resolution}
H.~K. B{\"{u}}ning, M.~Karpinski, and A.~Fl{\"{o}}gel.
\newblock {Resolution for quantified Boolean formulas}.
\newblock {\em Information and Computation}, 117(1):12--18, 1995.

\bibitem{Chen14-frontier}
H.~Chen.
\newblock The tractability frontier of graph-like first-order query sets.
\newblock In {\em Joint Meeting of the Twenty-Third {EACSL} Annual Conference
  on Computer Science Logic {(CSL)} and the Twenty-Ninth Annual {ACM/IEEE}
  Symposium on Logic in Computer Science (LICS), {CSL-LICS} '14, Vienna,
  Austria, July 14 - 18, 2014}, page~31, 2014.

\bibitem{ChenDalmau05-pebblegames}
H.~Chen and V.~Dalmau.
\newblock {From Pebble Games to Tractability: An Ambidextrous Consistency
  Algorithm for Quantified Constraint Satisfaction}.
\newblock In {\em Computer Science Logic 2005}, 2005.

\bibitem{ChenDalmau12-decomposingquantified}
H.~Chen and V.~Dalmau.
\newblock Decomposing quantified conjunctive (or disjunctive) formulas.
\newblock In {\em {LICS}}, 2012.

\bibitem{ChenDalmauGrussien13-ACandfriends}
H.~Chen, V.~Dalmau, and B.~Gru{\ss}ien.
\newblock Arc consistency and friends.
\newblock {\em J. Log. Comput.}, 23(1):87--108, 2013.

\bibitem{Egly12-sequent}
U.~Egly.
\newblock On sequent systems and resolution for qbfs.
\newblock In {\em SAT}, pages 100--113, 2012.

\bibitem{EglySeidlWoltran09-solver-nnf}
U.~Egly, M.~Seidl, and S.~Woltran.
\newblock A solver for {QBFs} in negation normal form.
\newblock {\em Constraints}, 14(1):38--79, 2009.

\bibitem{GiunchigliaNarizzanoTacchella06-clauseterm}
E.~Giunchiglia, M.~Narizzano, and A.~Tacchella.
\newblock Clause/term resolution and learning in the evaluation of quantified
  {B}oolean formulas.
\newblock {\em J. Artif. Intell. Res. (JAIR)}, 26:371--416, 2006.

\bibitem{JanotaMarquesSilva13-expansions}
M.~Janota and J.~Marques-Silva.
\newblock On propositional {QBF} expansions and {Q}-resolution.
\newblock In {\em SAT}, pages 67--82, 2013.

\bibitem{KolaitisVardi00-gametheoretic}
P.~G. Kolaitis and M.~Y. Vardi.
\newblock A game-theoretic approach to constraint satisfaction.
\newblock In {\em Proceedings of the 17th National Conference on AI}, pages
  175--181, 2000.

\bibitem{NarizzanoPeschieraPulinaTacchella09-qbfs}
M.~Narizzano, C.~Peschiera, L.~Pulina, and A.~Tacchella.
\newblock Evaluating and certifying qbfs: A comparison of state-of-the-art
  tools.
\newblock {\em AI Commun.}, 22(4):191--210, 2009.

\bibitem{Segerlind07-survey}
N.~Segerlind.
\newblock The complexity of propositional proofs.
\newblock {\em Bull. Symbolic Logic}, 13:417--626, 2007.

\bibitem{SlivovskySzeider12-pathdependencies}
F.~Slivovsky and S.~Szeider.
\newblock Computing resolution-path dependencies in linear time.
\newblock In {\em SAT}, pages 58--71, 2012.

\end{thebibliography}

\end{document}